\documentclass[1p]{elsarticle}
\usepackage{amsmath,amsthm,amsfonts,amssymb,bbm}

\def\CC{{\mathbb C}}

\def\A{{\mathcal A}}
\def\B{{\mathcal B}}
\def\C{{\mathcal C}}
\def\D{{\mathcal D}}
\def\E{{\mathcal E}}
\def\F{{\mathcal F}}

\def\Z{{\mathcal Z}}

\newtheorem{theorem}{Theorem}[section]

\newtheorem{corollary}[theorem]{Corollary}
\newtheorem{proposition}[theorem]{Proposition}
\newtheorem{lemma}[theorem]{Lemma}

\begin{document}

\title{On quasi-orthogonal systems of matrix algebras}

\author{Mih\'aly Weiner\fnref{fn}}
\ead{mweiner@renyi.hu}
\address{
University of Rome ``Tor Vergata'', Department of Mathematics\\
via della Ricerca Scientifica 1, 00133 Rome, Italy \\
(on leave from:
Alfr\'ed R\'enyi Institute of Mathematics\\
H-1364 Budapest, POB 127, Hungary)}
\fntext[fn]{Supported in part by the ERC Advanced Grant 227458 OACFT 
``Operator Algebras and Conformal Field Theory''.}

\begin{abstract}

In this work it is shown that certain interesting types of quasi-orthogonal
system of subalgebras (whose existence cannot be ruled out by 
the trivial necessary conditions) cannot exist. In particular, it is proved
that there is no quasi-orthogonal decomposition of  
$M_n(\CC)\otimes M_n(\CC)\equiv M_{n^2}(\CC)$
into a number of maximal abelian subalgebras and
factors isomorphic to $M_n(\CC)$ in which the number of
factors would be $1$ or $3$.

In addition, some new tools are introduced, too: for example,
a quantity $c(\A,\B)$,  which measures ``how close'' the
subalgebras $\A,\B \subset M_n(\CC)$ are to being quasi-orthogonal. 
It is shown that in the main cases of interest, $c(\A',\B')$ --- where 
$\A'$ and $\B'$ are the commutants 
of $\A$ and $\B$, respectively --- can be determined by 
$c(\A,\B)$ and the dimensions of $\A$ and $\B$. 
The corresponding formula is used to find some further
obstructions regarding quasi-orthogonal systems.
\end{abstract}


\maketitle

\section{Introduction}

Matrix-algebraic questions have often their roots in quantum information theory.
Mutually unbiased bases (MUB) are considered and investigated because of their relation for example to quantum state tomography \cite{wooters} or quantum cryptography \cite{security}.

A collection of MUB can be viewed as  a particular example of a quasi-orthogonal 
system of  subalgebras of $M_n(\CC)$ (in this work by {\it subalgebra} we shall always mean 
a $^*$-subalgebra containing $\mathbbm 1\in M_n(\CC)$; for definition and details on quasi-orthogonality see the next section). In algebraic terms, it is a quasi-orthogonal system of 
{\it maximal abelian} subalgebras (MASAs). 

Recently research has began in the non-commutative direction \cite{petz,petzkhan,opsz,ohno,pszw}, too. 
(Note that in some of these articles instead of  ``quasi-orthogonal'' the term ``complementary'' is used.)
Indeed, it should not be the commutativity of subalgebras deciding wether something 
deserves to be studied or not.  
From the point of view of quantum physics, the interesting quasi-orthogonal systems and 
decompositions are those that 
contain factors and MASAs only. (Factors are related 
to subsystems and MASAs are 
related to maximal precision measurements.) 

An example for a quantum physics motivated quasi-orthogonal system
which is composed of both abelian and non-abelian algebras  
is  the collection of following $3$ subalgebras
of $M_2(\CC)\otimes M_2(\CC)\equiv M_4(\CC)$ 
(i.e.\! the algebra of $2$ quantum bits):
$M_2(\CC)\otimes \mathbbm 1$ (the algebra associated
to the first qbit), $\mathbbm 1\otimes M_2(\CC)$ (the algebra associated
to the second qbit), and the maximal abelian subalgebra associated to
the so-called {\it Bell-basis} (which plays an important role 
e.g.\! in the protocol of {\it dense-coding}). 

Existential and constructional questions are already difficult in the
abelian case. We know many things when  the dimension is a power of 
a prime \cite{ivanovic,woofie}, but for example it is still a question, 
whether in $6$ dimensions there exists a collection of $7$ 
MUB or not \cite{h6,matemisi}.

Little is known when not all subalgebras are assumed to be maximal abelian. 
What are the existing constructions 
and established  obstructions (that is, reasons preventing the existence of 
certain such systems)? Of course there are some trivial necessary conditions (that will be 
discussed later).  Considering systems containing not only factors and MASAs, 
it is easy to see, that in general these conditions, alone,
cannot be also sufficient (see the example given in section \ref{sec:pre:decomp}). 
However, up to the knowledge of the author, previous to 
this work, nontrivial obstructions regarding ``interesting'' systems
were only found in very small dimensions 
(namely in dimension $4$, see \cite{petzkhan, opsz, pszw}), using --- in part --- some 
rather explicit calculations. Moreover, existing constructions such as the ones in
\cite{opsz,ohno} are usually carried out in prime-power dimensions, only. 
Thus there is a wide gap between constructions and obstructions where 
``anything could happen''.

The aim of this work is to shorten this gap. In particular, 
we shall exclude the existence of some interesting systems 
(and moreover, we shall do so not only in some low dimensions).

This paper is organized as follows. First, ---
partly for reasons of self-containment, partly for fixing notations --- a quick overview 
(including a presentation of the known results) is given about 
quasi-orthogonality, quasi-orthogonal systems 
and quasi-orthogonal decompositions. Though it is
well-known to experts, certain parts 
--- at least, up to the knowledge of the author --- have never been 
collected together. In particular, $3$ conditions will be singled out and 
listed as ``trivial necessary conditions'' of existence for a system.

Then in section \ref{sec:M_n} we consider decompositions of $M_n(\CC)\otimes M_n(\CC) = 
M_{n^2}(\CC)$. The tensorial product $M_n(\CC)\otimes M_n(\CC) \equiv M_{n^2}(\CC)$ 
appears in quantum physics when one deals with a bipartite system composed 
of two equivalent parts. Of course 
$M_{n^2}(\CC)$ has many subfactors isomorphic to $M_n(\CC)$ --- in physics 
such a subfactor may stand for a subsystem; for example $M_n(\CC)\otimes 
\mathbbm 1$ stands for the first part of the bipartite system. It seems 
therefore a natural question to investigate quasi-orthogonal 
decompositions of $M_{n^2}(\CC)$ into subfactors isomorphic to 
$M_n(\CC)$ and a number of MASAs. (As was mentioned,  
MASAs are related to maximal precision measurements.) 
 We shall show that there is no such decomposition in 
which there would be only $1$ factor (with the other algebras being 
maximal abelian) and neither there are decompositions with $3$ factors.
(Note that with $2$ factors there are decompositions, see 
\cite[Theorem 6]{pszw}, for example.) As far as
the author knows, this is the first example\footnote{In reality ---
though in a somewhat implicit manner --- another work \cite{lmsweiner} of the 
present author has already dealt with the case of a single factor; see the 
remark after corollary \ref{1factor}.
However, the non-existence of this kind of decomposition was never stated there ---
that paper had a different aim.}
for excluding the existence of  some ``interesting'' quasi-orthogonal systems 
(whose existence cannot be ruled out by the trivial necessary conditions) 
in an infinite sequence of higher and higher dimensions.

We shall deal with these cases using a recent result
\cite{ohnopetz}, by which if we replace 
each subalgebra in such a decomposition with its commutant,  we again get 
a quasi-orthogonal decomposition. 
However, this is something rather
particular:  in general, if two subalgebras are quasi-orthogonal, their
commutants will not remain so.  
To study the relation of the commutants,
in section \ref{sec:formula} for two subalgebras $\A,\B \subset M_n(\CC)$ 
we shall take the corresponding 
trace-preserving expectations $E_\A,E_\B$ and consider the quantity 
\begin{equation} 
c(\A,\B) : = {\rm Tr}(E_\A E_\B) 
\end{equation} 
where $E_\A E_\B$ is viewed as an $M_n(\CC) \rightarrow M_n(\CC)$ linear 
map (and hence its trace is well-defined). Then $c(\A,\B)\geq 1$ and 
equality holds if and only if $\A$ and $\B$ 
are quasi-orthogonal. 
Thus $c(\A,\B)$ measures how much $\A$ and $\B$ are (or: how 
much they are {\it not}) quasi-orthogonal. We shall prove, 
that if $\A$ and $\B$ satisfy a certain homogenity condition (which is 
always satisfied, if they are 
factors or maximal abelian subalgebras) then 
\begin{equation}
c(\A',\B') = \frac{n^2}{{\rm dim}(\A) {\rm dim}(\B)} c(\A,\B).
\end{equation}
Finally, in the last section we shall  show in some concrete examples how the derived 
formula can be used to generalize our earlier arguments and thus to exclude the existence of some further quasi-orthogonal systems. 
In some sense our examples will fall ``close" to the cases dealt 
with in section \ref{sec:M_n}.  However, in contrast to those cases, here the commutants will not
remain (exactly) quasi-orthogonal; so instead of  ``exact" statements we shall rely on 
our quantitative formula.

\section{Preliminaries}
\label{sec:preliminaries}

\subsection{Quasi-orthogonality}

There is a natural scalar product on $M_n(\CC)$ (the so-called
{\it Hilbert-Schmidt} scalar product) defined by the formula
\begin{equation}
\langle A,B\rangle = {\rm Tr}(A^*B)\;\;\;\;\;\;\;\;\; (A,B\in M_n(\CC)).
\end{equation}
Thus if $\A\subset M_n(\CC)$ is a linear subspace, it is 
meaningful to consider the ortho-projection $E_\A$ onto $\A$. When $\A$ 
is actually a $*$-subalgebra containing $\mathbbm 1 \in M_n(\CC)$ (or in 
short: a subalgebra), $E_\A$ coincides with the so-called 
{\it trace-preserving conditional expectation} onto $\A$.

Two subalgebras $\A,\B\subset M_n(\CC)$, as linear subspaces, cannot be 
orthogonal, since $\A\cap \B\neq \{0\}$ as $\mathbbm 1\in \A\cap \B$. At 
most, the subspaces $\A\cap \{\mathbbm 1\}^\perp$ and $\B\cap \{\mathbbm
1\}^\perp$ can be orthogonal, in which case we say that $\A$ and $\B$ 
are {\bf quasi-orthogonal}. 

Note also that $A\in M_n(\CC)$ is orthogonal to $\mathbbm 1$ if and only 
if ${\rm Tr}(A)=0$ and so the subspace $\A\cap \{\mathbbm 1\}^\perp$ is 
simply the ``traceless part'' of $\A$. In other words, $\A$ and $\B$ are 
quasi-orthogonal if and only if their traceless parts are orthogonal.

For an $X\in M_n(\CC)$ denote its traceless part by $X_0$; that is,
\begin{equation}
X_0=X-\tau(X)\mathbbm 1
\end{equation}
where $\tau = \frac{1}{n}{\rm Tr}$ is the {\bf normalized trace}. (Note 
that the normalization is done in such a way that $\tau(\mathbbm 1)=1$.) 
Then the traceless parts $A_0,B_0$ of $A,B\in M_n(\CC)$ are orthogonal
if and only if 
\begin{equation}
0=\tau(A_0^*B_0)=\tau((A^*-\overline{\tau(A)}\mathbbm 
1)(B-\tau(B)\mathbbm 1) = \tau(A^*B)-\overline{\tau(A)}\tau(B),
\end{equation}
that is, if and only if $\tau(A^*B)=\tau(A^*)\tau(B)$. So, since if 
$A$ is an element of the subalgebra, then so is $A^*$, we have that
two subalgebras $\A,\B$ of $M_n(\CC)$ are quasi-orthogonal if and only if 
for all $A\in \A$ and $B\in \B$,
\begin{equation}\label{product-trace}
\tau(A B) = \tau(A)\tau(B).
\end{equation}

\subsection{Factors, abelian subalgebras and MUB}
For any subalgebra $\A\subset M_n(\CC)$ one can consider its {\bf 
commutant}
\begin{equation}
\A'\equiv \{X\in M_n(\CC)| \,\forall A\in \A: \, AX-XA = 0\}
\end{equation}
which is again a subalgebra. One has that the {\bf second commutant}
$\A''\equiv (\A')'=\A$. A subalgebra $\A$ whose center 
\begin{equation}
\Z(\A)=\A\cap \A'
\end{equation}
is trivial (i.e.\! such that $\Z(\A)=\CC\mathbbm 1$) is called a {\bf 
factor}. If $\A\subset M_n(\CC)$ is a factor, then there exist $j,k$ natural 
numbers such that $jk=n$, and that up to unitary equivalence, $\A$ is of 
the form
\begin{equation}
\A=M_j(\CC)\otimes \mathbbm 1 \equiv \{ A\otimes \mathbbm 1 | A\in 
M_j(\CC)\} \subset M_j(\CC)\otimes M_k(\CC)\equiv M_{jk}(\CC).
\end{equation}
Then $\A'=\mathbbm 1 \otimes M_k(\CC)$ and so if $\A$ is factor, then 
$\A$ and $\A'$ are always quasi-orthogonal; this follows easily from the 
trace-criterion (\ref{product-trace}) and the fact that
\begin{equation}
{\rm Tr}(A\otimes B) = {\rm Tr}(A) {\rm Tr}(B)
\end{equation}
for all $A\in M_j(\CC)$ and $B\in M_k(\CC)$.

Another example of quasi-orthogonal subalgebras comes from mutually 
unbiased bases. Two orthonormed bases $\E=(\mathbf e_1, \ldots, \mathbf
e_n)$ and $\F=(\mathbf f_1, \ldots, \mathbf f_n)$ in $\CC^n$ such that
\begin{equation}
|\langle \mathbf e_k,\mathbf f_j \rangle| = {\rm constant} =
\frac{1}{\sqrt{n}}
\end{equation}
for all $k,j=1,\ldots, n$, are said to be {\bf mutually unbiased}, or in 
short, $\E$ and $\F$ is a pair of MUB. 

Clearly, unbiasedness does not depend on the order of vectors in $\E$ 
and $\F$, nor on their ``phase factors''. (That 
is, the MUB property is not disturbed by replacing a basis vector $\mathbf 
v$ by $\lambda \mathbf v$, where $\lambda\in \CC, |\lambda|=1$.) For 
this reason, one often associates subalgebras to 
these bases (which do not depend on the order of vectors and their 
phases) and then works with them rather than with the actual bases.

Let us see how can we assign a subalgebra to an orthonormed basis $\E$.
For a vector $\mathbf v\neq 0$, denote the ortho-projection onto the 
one-dimensional subspace $\CC\mathbf v$ by $P_{\mathbf v}$. Then the 
linear subspace of $M_n(\CC)$
\begin{equation} 
\A_\E \equiv {\rm Span}\{\, P_{\mathbf e_j}\, |j=1,\ldots, n\} 
\end{equation} 
is actually a subalgebra. Infact it is a {\bf maximal abelian subalgebra} 
(in short: a MASA), and every MASA of $M_n(\CC)$ is of this form. 

Elementary calculation shows that if $\mathbf v, \mathbf w$ are vectors of 
unit length then 
\begin{equation}
{\rm Tr}(P_{\mathbf v}P_{\mathbf w}) = 
|\langle \mathbf v,\mathbf w\rangle |^2. 
\end{equation}
Hence by an application of the trace-criterion (\ref{product-trace}) one 
has that $\E$ and $\F$ is a pair of MUB if and only if the associated 
maximal abelian subalgebras $\A_\E$ and $\A_\F$ are quasi-orthogonal.

A famous question concerning MUB is: how many orthonormed 
bases can be given in $n$ dimensions in such a way that any two of 
the given collection is a MUB?  There is a simple bound concerning this 
maximum number --- which we shall 
denote by $N(n)$ --- namely, that if $n>1$ then $N(n)\leq n+1$. Let us 
recall now how this bound can be obtained by a use of the above 
introduced subalgebras.

The traceless part of $M_n(\CC)$ is $n^2-1$ dimensional, whereas the 
traceless part of a maximal abelian subalgebra of $M_n(\CC)$ is
$n-1$ dimensional. If $n>1$, then at most 
\begin{equation}
\frac{n^2-1}{n-1}=n+1
\end{equation}
$n-1$-dimensional orthogonal subspaces can be fitted in an $n^2-1$ 
dimensional space, implying that for $n>1$ we have $N(n)\leq n+1$.

It is known by construction \cite{ivanovic,woofie} that if $n$ is a power of a prime, 
then $N(n) = n+1$. However, apart from 
$n=p^k$ (where $p$ is a prime), there is no other dimension $n>1$ in 
which the value of $N(n)$ would be known. In particular, 
already the value of $N(6)$ is an open question with a long literature on 
its own; see e.g.\! \cite{h6,matemisi} . All we know is that $3\leq 
N(6)\leq 7$ with numerical evidence \cite{numerical} indicating 
that $N(6)$ is actually $3$.

\subsection{Quasi-orthogonal systems and decompositions}
\label{sec:pre:decomp}

A collection of pairwise quasi-orthogonal subalgebras 
$\A_1,\A_2,\ldots$ of $M_n(\CC)$ is said to form a {\bf 
quasi-orthogonal system} in $M_n(\CC)$. 
If in addition the given subalgebras linearly span the full space 
$M_n(\CC)$, we say that the collection is a {\bf quasi-orthogonal 
decomposition} of $M_n(\CC)$.

Suppose we are looking for a quasi-orthogonal system in
$M_n(\CC)$ (or quasi-orthogonal {\it decomposition} of $M_n(\CC)$) 
$\A_1,\ldots, \A_k$ such that $\A_j$ is isomorphic to $\B_j\; 
(j=1,\ldots,k)$, where $\B_1,\ldots,\B_k$ are given matrix 
algebras. For example, we may look for a quasi-orthogonal system in which 
each algebra is a MASA --- as is the case when we want to find a 
collection of MUB --- or, motivated by the study of ``quantum bits'' we may look for a system consisting of subalgebras all isomorphic to $M_2(\CC)$ --- as is investigated in 
\cite{petzkhan,opsz}. 

What can we say about the existence of a specific system? Some necessary 
conditions are easy to establish. In particular, the following three 
will be referred as the ``trivial neccessary conditions'' 
for the existence of a specific quasi-orthogonal system in $M_n(\CC)$ (or: 
quasi-orthogonal decomposition of $M_n(\CC)$).

\begin{itemize}
\item[(1)] $M_n(\CC)$ must contain {\it some} subalgebras 
$\A_1,\ldots, \A_k$ isomorphic 
to the given algebras $\B_1,\ldots, \B_k$, respectively,
\item[(2)] the product ${\rm dim}(\B_i) {\rm dim}(\B_j)\leq n^2$ 
for all $1\leq i<j \leq k$,
\item[(3)] $\sum_{j=1}^k ({\rm dim}(\B_j)-1) \leq n^2-1$ and a 
corresponding quasi-orthogonal system is a quasi-orthogonal {\it 
decomposition} if and only if in the above formula equality holds. 
\end{itemize}

The first condition does not require too much explanation. Nevertheless, 
it rules out the existence of various quasi-orthogonal systems. For 
example, can we have a quasi-orthogonal system in $M_5(\CC)$ consisting of 
$3$ subalgbebras each of which is isomorphic to $M_2(\CC)$? Clearly no: 
simply, $M_5(\CC)$ does not contain any subalgebra that would be 
isomorphic to $M_2(\CC)$ since $2$ does not divide $5$.

The second condition, at first sight, is perhaps less evident; 
let us see now why is it necessary. Suppose $\A$ and $\B$ are 
quasi-orthogonal subalgebras of $M_n(\CC)$. Let $A_1,\ldots A_{d_\A}$ 
and $B_1,\ldots, B_{d_\B}$ be orthonormed bases in $\A$ and $\B$
(with $d_\A,d_\B$ standing for the dimensions of $\A$ and $\B$), 
respectively. Then, by definition of the (Hilbert-Schmidt) scalar 
product and by the trace property (\ref{product-trace}) we have that
\begin{eqnarray}\label{product-basis}
\nonumber
n \langle A_iB_j, A_{i'}B_{j'}\rangle &=& 
n^2 \tau( (A_iB_j)^* A_{i'}B_{j'})=
n^2 \tau(A_i^*A_{i'}B_{j'}B_j^*) 
= n^2 \tau(A_i^*A_{i'})\tau(B_{j'}B_j^*) \\
&=&
n^2 \tau(A_i^*A_{i'})\tau(B_j^*B_{j'})=
\langle A_i, A_{i'}\rangle\, 
\langle B_j, B_{j'}\rangle, 
\end{eqnarray}
showing that $\sqrt{n}A_iB_j\; (i=1,\ldots,d_\A; \, j=1,\ldots,d_\B)$ is an 
orthonormed system in $M_n(\CC)$. Hence the number of members in this system 
must be less or equal than the dimension of the full space $M_n(\CC)$; that is,
$d_\A d_\B\leq n^2$.

The third condition is necessary simply because if $\A_1,\ldots, \A_k$ are
quasi-orthogonal, then their traceless parts are orthogonal 
subspaces in the traceless part of $M_n(\CC)$. We argue exactly 
like we did at discussing the maximum number of MUB
(which, for us, is just a particular case): we have $k$ orthogonal 
subspaces of dimensions ${\rm dim}(\A_j)-1\; (j=1,\ldots,k)$ in a
${\rm dim}(M_n(\CC))-1=n^2-1$ dimensional space, implying the claimed 
inequality. Moreover, the subspaces span the full space (i.e.\! 
we have a quasi-orthogonal {\it decomposition}) if and only if the
dimensions add up exactly to $n^2-1$.

So these conditions are necessary for existence. But are they also sufficient?
The answer, in general, is not.
\bigskip 

\noindent
{\it Example.}
Can we find a quasi-orthogonal  system in $M_{2n}(\CC)$ consisting of 
an abelian subalgebra $\A$ of 
dimension $n+1$ and a factor $\B$ isomorphic to $M_n(\CC)$? 
If $n>2$, the answer is: not. Indeed, assume by contradiction that $\A,\B$
is such a pair. Let $P_1,\ldots ,P_{n+1}$ be the minimal projections of $\A$.
Since we are in a $2n$-dimensional space, at least one of them is a projection onto 
a one-dimensional space. So suppose $P_k$ is the orthogonal projection onto the subspace 
generated by the unit-length vector $x$. Then by the trace property (implied by quasi-orthogonality)
\begin{equation}
\langle x, Bx\rangle = {\rm Tr}(P_k B) = \frac{1}{2n} {\rm Tr}(P_k){\rm Tr}(B) = \frac{1}{2n}{\rm Tr}(B)
\end{equation}
for all $B\in \B$. This shows that the linear map $B\mapsto Bx$ is injective
on $\B$. Indeed, if $Bx=0$ then $0=\|Bx\|^2= \langle Bx,Bx\rangle =
\langle x, B^*Bx\rangle$, which by the above equation would mean that
${\rm Tr}(B^*B)=0$, implying that $B=0$. However, this is a contradiction, as
the dimension of $\B$ is bigger than the dimension of the full 
space: $n^2>2n$. Yet the listed necessary conditions would allow the existence of
such a system. Indeed, the first condition is trivially satisfied, the second is satisfied 
as ${\rm dim}(\A) {\rm dim}(\B) = (n+1)n^2< (2n)^2$, whereas the third is satisfied
since $({\rm dim}(A)-1)+({\rm dim}(\B)-1) = n+n^2-1 < (2n)^2-1$.
\smallskip

 \noindent
Since our motivation is quantum information theory, we are mainly interested 
by quasi-orthogonal systems formed by maximal abelian subalgebras and factors. 
For such systems it is somewhat more difficult to show that the trivial necessary 
conditions are not also sufficient. Let us continue now by discussing 
the known examples of  ``interesing'' systems.

The trivial necessary conditions allow the existence of a 
quasi-orthogonal system in $M_n(\CC)$ composed of $k$ MASAs 
as long as $k\leq n+1$ (see the third condition). Moreover,  
a quasi-orthogonal system composed of exactly $n+1$ MASAs
would give a quasi-orthogonal decomposition of $M_n(\CC)$. 
As was mentioned, the existence of such systems is a popular research 
theme (though the problem is usualy considered rather in terms of 
MUB than MASA), and little is known when $n$ is not a power of a prime.

Another, more recent problem is to find a system of 
quasi-orthogonal subalgebras in $M_{2^k}(\CC)$
in which all subalgebras are isomorphic to $M_2(\CC)$. Here there 
is a more direct motivation: a quantum bit, in some sense, is 
a subalgebra isomorphic to $M_2(\CC)$, whereas the full algebra
$M_{2^k}(\CC)\simeq
M_2(\CC)\otimes M_2(\CC)\otimes\ldots $ is used in the description of
the register of a quantum computer containing $k$ quantum bits.

In this case too, the first two trivial necessary conditions are 
automatically satisfied, whereas the third one says that 
such a system can consists of at most $(2^{2k}-1) / 3 =:S(k)$ 
subalgebras. Again, exactly $S(k)$ such subalgebras would give a 
quasi-orthogonal decomposition. (It is easy to see that $S(k)$ 
is an integer.) In \cite{opsz} $S(k)-1$ such subalgebras are 
presented by a construction using induction on $k$. For
$k>2$ it is not known whether the construction is {\it optimal}; 
that is, whether the upper bound $S(k)$ could be realized or 
not. However, it is proved \cite{petzkhan} that for $k=2$ --- i.e.\! in 
$M_4(\CC)$ --- the construction is indeed optimal: there is no 
quasi-orthogonal system consisting of $S(2)=5$ subalgebras 
isomorphic to $M_2(\CC)$. This shows that the listed trivial 
necessary conditions, even in the special case of our interest, are not always 
sufficient, too. (As far as the author of this 
work knows, this was the first example of an ``interesting'' quasi-orthogonal 
system satisfying the trivial conditions, whose existence was 
{\it disproved}.)

The case of $M_4(\CC)$ has received quite a bit of attention 
\cite{petzkhan,opsz,pszw}. Indeed, this is the smallest dimension in which 
--- at least from our point of view --- something nontrivial 
is happening. As we are interested by factors and MASAs, 
let us consider a quasi-orthogonal decomposition
of $M_4(\CC)$ consisting of a collection of MASAs
(so subalgebras isomorphic to $\CC^4$) and proper subfactors (so 
subalgebras isomorphic to $M_2(\CC)$). Again, the 
first two trivial necessary conditions are automatically satisfied, 
whereas dimension counting (third condition) says that for such a 
decomposition we need $5$ subalgebras. The trivial necessary conditions
do not give anything more. However, in \cite{pszw} it was proved that such 
a decomposition exists if and only if an even number of these $5$ 
subalgebras are factors. So for example one can construct such a 
decomposition with $3$ factors and $2$ MASAs, but 
not with $2$ factors and $3$ MASAs. This again shows that 
the trivial necessary conditions are not always sufficient, too.

The problem with quasi-orthogonal copies of $M_2(\CC)$ in $M_{2^k}(\CC)$
can be also generalized in the sense that one may look for 
quasi-orthogonal copies of $M_n(\CC)$ in $M_{n^k}(\CC)$. If $n$ is 
a power of a non-even prime, then $M_{n^k}(\CC)$ admits a 
quasi-orthogonal decomposition into subalgebras isomorphic to 
$M_n(\CC)$. (As was mentioned, the same does {\it not} hold for $n=k=2$.)
The proof is constructional and relies on the existence of finite fields 
and in some sense it is carried out in a similar manner to 
the construction of $n+1$ MUB in dimension $n=p^\alpha$ (where $p>2$ is a 
prime and $\alpha$ is a natural number).

\section{Decompositions of $M_n(\CC)\otimes 
M_n(\CC) \equiv M_{n^2}(\CC)$}
\label{sec:M_n}

We shall now consider quasi-orthogonal 
decompositions of $M_{n^2}(\CC)$ into subfactors isomorphic to 
$M_n(\CC)$ and a number of MASAs. 
Such decompositions of $M_2(\CC)\otimes M_2(\CC)\equiv 
M_4(\CC)$ are well studied in \cite{pszw}. However, there the achieved 
results relay on explicit calculations carried out in $4$ dimensions. What 
can we do in higher dimensions?
Note that the trivial necessary conditions do not rule out the 
existence of a decomposition of the mentioned type;
 all they say that such decompositions must consists of
\begin{equation}
\frac{n^4-1}{n^2-1} = n^2+1
\end{equation} 
subalgebras (since both a maximal abelian subalgebra of $M_{n^2}(\CC)$ 
and the factor $M_n(\CC)$ is $n^2$-dimensional).

Decomposition into MASAs is of course interesting, 
but it is known to be a hard question which is usualy studied in terms of 
MUB and it is out of the scope of this article. Actually, there is a certain mathematical (or 
more precisely: operator algebraic) advantage of having not only MASAs: it is often 
helpful to consider the commutant of a subalgebra. (The commutant 
of a MASA is itself, so it does not give anything ``new''.) 
Infact, the result in \cite{petzkhan} concerning quasi-orthogonal 
copies of $M_2(\CC)$ in $M_4(\CC)$ is achieved exactly by considering 
commutants.

In \cite{ohnopetz} an important result is deduced about the quasi-orthogonality 
of commutants. We shall 
now recall this result (stating it in a sightly different form).

\begin{lemma}
Let $\A_1$ and $\A_2$ be quasi-orthogonal subalgebras of $M_n(\CC)$. Then 
the commutants $\A_1'$ and $\A_2'$ are quasi-orthogonal if and only if 
${\rm dim}(\A_1)\, {\rm dim}(A_2)=n^2$.
\end{lemma}
\begin{proof}
The observation established by calculation (\ref{product-basis}) 
shows that the required equality holds if and only if 
the set $\{A_1A_2|\,A_1\in\A_1,A_2\in\A_2\}$ spans
$M_n(\CC)$. Hence our lemma is a simple reformulation of one of the 
claims of the original statement \cite[Prop.\! 2]{ohnopetz} 
\end{proof}

\begin{corollary}\label{1factor}
There is no quasi-orthogonal decomposition of $M_{n^2}(\CC)$ into maximal 
abelian subalgebras and a (single) factor isomorphic to $M_n(\CC)$.
\end{corollary}
\begin{proof}
Suppose the maximal abelian algebras $\A_1,\ldots,\A_{n^2}$ together 
with the factor $\B$ form such a decomposition. Then, since both 
${\rm dim}(\B)={\rm dim}(M_n(\CC))=n^2$
and also the dimension of a maximal abelian subalgebra of $M_{n^2}(\CC)$ 
is $n^2$, by the previous lemma we have that $\B'$ is quasi-orthogonal 
to $\A'_k=\A_k\; (k=1,\ldots,n^2)$. But since $\B$ is a factor, $\B'$ is 
also quasi-orthogonal to $\B$. Hence $\B'$ should be quasi-orthogonal to 
each member of a quasi-orthogonal decomposition, implying that $\B'$ 
should be equal to $\CC \mathbbm 1$ and in turn, that $\B=\B''$ should be 
the full matrix algebra $M_{n^2}(\CC)$ (which is clearly a contradiction).
\end{proof}

\noindent
{Remark.} In \cite{lmsweiner} the author of the present work has shown that
if $\A_1,\ldots,\A_d$ is a system of $d$ 
MASAs in $M_d(\CC)$, then any pair of elements in 
the orthogonal subspace $(\A_1+\ldots +\A_d)^\perp$ must commute. In particular, 
if $\A_1,\ldots,\A_d,\B$ is a quasi-orthogonal system in $M_d(\CC)$ where
$\A_1,\ldots,\A_d$ are MASAs, then $\B$ must be a 
commutative algebra. This is of course a much stronger affirmation than the
above corollary. However, that article uses a much longer proof 
and the method presented here has the further advantage that 
--- as we shall shortly see --- it can be applied to cases when the number of
MASAs is less than $d$. In any case, the aim of the cited
work was to study mutually unbiased bases
(and not quasi-orthogonal decompositions in general); 
the nonexistence of the above discussed system was not
stated explicitly there.
\bigskip

\noindent
Now how about decompositions of $M_{n^2}(\CC)$ into MASAs
and {\it two} factors isomorphic to $M_n(\CC)$? Such 
decompositions, in general, cannot be ruled out. Indeed, as was mentioned, 
in \cite{pszw} the case of $n=2$ was treated and in particular an example was 
given for such a decomposition. Moreover, it was shown 
that there are no decompositions of $M_4(\CC)$ into 
MASAs and factors isomorphic to $M_2(\CC)$ in which the number of 
factors would be $1,3$ or $5$. For general $n>1$, we shall now prove that 
there is no decompositions of $M_{n^2}(\CC)$ into MASAs
and factors isomorphic to $M_n(\CC)$ in which the number 
of factors is $3$ (and we have already seen that nor it can be $1$). We 
will need some preparatory steps.
\begin{lemma}
\label{A1A2}
Let $\A_1$ and $\A_2$ be quasi-orthogonal subalgebras of $M_n(\CC)$ and
$A_1\in \A_1$ and $A_2\in \A_2$ two traceless operators. Then
$A_jA_k \in \A_1+\A_2$ if $j=k$ whereas if $j\neq k$ then
$A_jA_k$ is orthogonal to the subspace $\A_1+\A_2$.
\end{lemma}

\begin{proof}
Apart from trivial affirmations, all we have to check that 
is that the ``cross-terms'' $A_jA_k$ (where $j\neq k$) are orthogonal
to the subspace $\A_1+\A_2$. If $X\in \A_1$, then 
\begin{equation}
\langle X, A_1A_2\rangle = 
{\rm Tr}(X^*A_1A_2)=
{\rm Tr}((A_1^* X)^*W_2)
=
\langle (A_1^* X), A_2\rangle = 0
\end{equation}
since $(A_1^* X)\in \A_1$ whereas $A_2$ is a traceless operator in $\A_2$
which is supposed to be quasi-orthogonal to $\A_1$. If $X\in \A_2$ then using the 
invariance of trace under cyclic permutations we still get that 
\begin{equation}
\langle X, A_1A_2\rangle = 
{\rm Tr}(X^*A_1A_2)=
{\rm Tr}(A_1A_2 X^*)
=
\langle A_1^*, (A_2 X^*)\rangle = 0
\end{equation}
as the traceless element $A_1^*$ of $\A_1$ is orthogonal to any element of $\A_2$ (and in particular,
to $A_2X^*$). The rest (the orthogonality of the other cross-term: $A_2A_1$) follows by symmetry of the argument.
\end{proof}
\begin{lemma}
\label{BcapA=1}
Let $\A_1$ and $\A_2$ be quasi-orthogonal subalgebras of $M_n(\CC)$ and 
suppose that $\B$ is a third subalgebra of $M_n(\CC)$ such that $\B\subset \A_1+\A_2$.
Then either $\B\subset \A_j$ for some $j=1,2$ or $\B\cap \A_1 = \B\cap \A_2 = \CC \mathbbm 1$.
\end{lemma}

\begin{proof}
Suppose there exists a $B\in \B$ which is neither in $\A_1$ nor in $\A_2$.  Then its traceless part 
\begin{equation}
B_0 = B - \tau(B)\mathbbm 1 = B- ({\rm Tr}(B)/n)\mathbbm 1
\end{equation}
is still an element of $\B$ which is neither in $\A_1$ nor in $\A_2$, so 
\begin{equation}
B_0=B_1 + B_2
\end{equation}
for some $B_j\in \A_j$ nonzero traceless operators $(j=1,2)$.
If  $X\in \B \cap \A_1$ then again its traceless part  $X_0$ is still in the intersection $\B\cap \A_1$.
Thus $X_0 B_1\in \A_1$ whereas 
by the previous lemma $X_0B_2$ is orthogonal to $\A_1+\A_2$ since
$X_0\in \A_1$.  On the other hand, as $X_0\in \B$, we have that
$X_0 B_1 + X_0 B_2 = X_0 B_0 \in \B\subset \A_1+ \A_2$. These two things imply that
$X_0 B_2 = 0$. 

Of course the fact that $B_2\neq 0$ is not enough for showing that $X_0=0$.
However, the argument presented at equation (\ref{product-basis}) shows that --- apart from a
factor depending on the dimension $n$ --- the {\it trace-norm} of a product of two elements belonging to two quasi-orthogonal subalgebras is simply the product of norms. Hence in our case $X_0 B_2 = 0$ 
actually {\it does} imply that one of the terms in the product must be zero, and so that $X_0=0$. 
Thus the arbitrary element $X$ of the intersection $\B\cap \A_1$ is a multiple of the identity; 
that is $\B\cap \A_1 = \CC\mathbbm 1$. The rest of the claim follows by repeating the argument with 
$\A_1$ and $\A_2$ exchanged.
\end{proof}

\begin{lemma}
\label{anti-comm}
Let $\A_1$ and $\A_2$ be quasi-orthogonal subalgebras of $M_n(\CC)$ and 
suppose that $A_j\in \A_j$ are traceless operators ($j=1,2$) 
such that $A^2\in \A_1+\A_2$ where $A=A_1+A_2$. Then $A_1$ and $A_2$ must 
anti-commute.
\end{lemma}
\begin{proof}
The claim is evident because by the previous lemma, in the expansion
$A^2= A_1^2+A_2^2 + A_1A_2+A_2 A_1$,
the first two terms are in $\A_1+\A_2$ whereas the last two terms are orthogonal to this subspace.
\end{proof}

\begin{theorem}
Let $\A_1$ and $\A_2$ be quasi-orthogonal subalgebras of $M_n(\CC)$ and 
suppose that $\B$ is a third subalgebra of $M_n(\CC)$ such that $\B\subset \A_1+\A_2$.
Then either $\B\subset \A_j$ for some $j=1,2$ or $\B\simeq \CC^2$.
\end{theorem}
\begin{proof}
Assume by contradiction that $\B$  is not included in neither of the two given quasi-orthogonal subalgebras, but $\B$ is not isomorphic to $\CC^2$. Then ${\rm dim}(B)>2$ (since up to isomorphism, 
there is only one two dimensional star-algebra: $\CC^2$) and so the
traceless part of  $\B$,
\begin{equation}
\B_0:= \B\cap \{\mathbbm 1\}^\perp
\end{equation}
is at least 2-dimensional. 

Let $E_j$ be the trace-preserving expectation onto $\A_j$ for $j=1,2$.  For any traceless element
$X\in \A_1+\A_2$ we have that $X=E_1(X)+E_2(X)$, so $\B_0\subset E_1(\B_0)+E_2(\B_0)$. 
Moreover,  this inclusion cannot be an equality, since in that case $\B_0$ would nontrivially 
intersect $\A_1$ or $\A_2$, contradicting to our previous lemma. 
Thus at least one out of the subspaces: $E_1(\B_0), E_2(\B_0)$ must be more than
1-dimensional; we may assume that ${\rm dim}(E_2(\B_0))>1$.

Let $B\in \B$ be a traceless self-adjoint element. Then $B=B_1+B_2$ where $B_j=E_j(B)$ $(j=1,2)$, 
and since  ${\rm dim}(E_2(\B_0))>1$, there exists a $\tilde{B}\in \B_0$ such that 
$\tilde{B}_2=E_2(\tilde{B})$ is nonzero and orthogonal to $B_1$.  As $\B$ is an algebra, we have that 
$B\tilde{B}\in \B\subset \A_1+\A_2$. But $B\tilde{B} = B_1\tilde{B}_1+ B_2\tilde{B}_2 + 
B_1\tilde{B_2}+ B_2\tilde{B}_1$, and according to lemma \ref{A1A2}, the first two 
terms in this sum are in $\A_1+\A_2$ whereas the last two terms are orthogonal to this subspace, so 
actually $B_1\tilde{B_2}+ B_2\tilde{B}_1 = 0$. On the other hand, by using the product-property 
(\ref{product-trace}), the anti-commutativity of $B_1$ and $B_2$ (assured by lemma \ref{anti-comm}),
 and the fact that $B_2= E_2(B)$ is self-adjoint (as so is $B$), we have that 
\begin{eqnarray}
\nonumber
\langle B_1\tilde{B_2}, B_2\tilde{B}_1 \rangle &=& 
{\rm Tr}((B_1\tilde{B_2})^*B_2\tilde{B}_1) = 
{\rm Tr}(\tilde{B_2}^*B_1 B_2\tilde{B}_1) 
= - {\rm Tr}(\tilde{B_2}^*B_2 B_1\tilde{B}_1 )
\\
&=& n \tau(\tilde{B_2}^*B_2 B_1\tilde{B}_1 )
=  
 n \tau(\tilde{B_2}^*B_2) \tau(B_1\tilde{B}_1) = 0,
  \end{eqnarray}
since by assumption $0=\langle \tilde{B_2}, B_2\rangle= {\rm Tr}(\tilde{B_2}^*B_2)= 
n \tau(\tilde{B_2}^*B_2)$. Thus $B_1\tilde{B_2}=B_2\tilde{B}_1=0$ since they are orthogonal but 
their sum is zero. As $\tilde{B}_2\neq 0$, this implies (by the argument already explained towards the end of the proof of lemma \ref{BcapA=1}) that $\B_1=0$. That is, $B=B_2$ is actually an element of $\A_2$. But our assumption, together with lemma \ref{BcapA=1} imply that 
 $\B\cap \A_2 = \CC\mathbbm 1$. It should then further follow that $B=0$; that is, we have shown that any self-adjoint traceless element in $\B$ is zero and hence that $\B=\CC\mathbbm 1$ which contradicts to the assumption that $\B$ is {\it not} a subalgebra of $\A_1$ or $\A_2$. 
 \end{proof}

\begin{theorem}
There are no quasi-orthogonal decompositions of $M_{n^2}(\CC)$ into 
maximal abelian subalgebras and factors isomorphic to $M_n(\CC)$ in which 
the number of factors would be $1$ or $3$.
\end{theorem} 

\begin{proof}
We have already proved the case in which the number of factors is $1$, so now assume by 
contradiction that we have a quasi-orthogonal decomposition containing three factors: $\B_1,\B_2,\B_3$
(all isomorphic to $M_n(\CC)$) and $n^2-2$ MASAs $\A_1,\ldots, \A_{n^2-2}$. Then, as was already noted and applied, considering the commutants: $\B_1',\B_2',\B_3',\A_1,\ldots, 
\A_{n^2-2}$
(where we have used that the commutant of a MASA is itself), we
still have a quasi-orthogonal decomposition.  Thus, $\B_1'$ is quasi-orthogonal to both $\B_1$ (since it is a factor) and the algebras $\A_1,\ldots, \A_{n^2-2}$ and hence $\B_1\subset 
\B_2+\B_3$. By our previous theorem it then follows that $\B_1'$ is either equal to $\B_2$ or to 
$\B_3$. Repeating our argument for $\B_2$ and $\B_3$, we see that the $\B'_j=\B_{\sigma(j)}$
for some $\sigma:\{1,2,3\}\to\{1,2,3\}$ such that:
\begin{itemize}
\item $\sigma^2= {\rm id}$ (since the second commutant gives back the original algebra),
\item $\sigma$ has no fixed points ($\B'_j\neq \B_j$ as $\B_j$ is not a MASA). 
\end{itemize}
However, these two properties are evidently contradicting.
\end{proof}

\section{The trace formula}
\label{sec:formula}

Suppose $P$ and $Q$ are the ortho-projections onto the subspaces $N$ and $K$, 
respectively. By elementary arguments involving traces and positive 
operators, one has that ${\rm Tr}(PQ)$ is a nonnegative real number,
\begin{equation}
{\rm dim}(N\cap K) \leq {\rm Tr}(PQ)\leq {\rm min}\{{\rm dim}(N), \, {\rm dim}(K)\},
\end{equation}
and moreover that ${\rm dim}(N\cap K) = {\rm Tr}(PQ)$ if and only if 
$N\cap (N\cap K)^\perp$ and $K\cap (N\cap K)^\perp$  are orthogonal.
Thus we may say that the nonnegative number ${\rm Tr}(PQ)-{\rm dim}(N\cap K)$ 
measures ``how much'' the subspaces $N\cap (N\cap K)^\perp$ and 
$K\cap (N\cap K)^\perp$  are {\it not} orthogonal. This number is zero if and only 
if they are orthogonal, and in some sense the bigger it is, the further away they 
are from orthogonality. Let us see now what this has to do with quasi-orthogonal subalgebras.

A subalgebra 
$\A\subset M_n(\CC)$ is in particular a linear subspace. As was 
discussed, $M_n(\CC)$ has a natural scalar product, so it is meaningful to talk 
about orthogonality. Thus we may consider the ortho-projection $E_\A$ onto $\A$. 
Note that this map is usually referred as the {\it unique trace-preserving 
expectation} onto $\A$. For two subalgebras $\A,\B \subset M_n(\CC)$ we shall now 
introduce the quantity
\begin{equation}
c(\A,\B):={\rm Tr}(E_\A E_\B).
\end{equation}
Note that here ${\rm Tr}$ is the trace of the set of linear operators {\it acting} 
on $M_n(\CC)$, and not the trace of $M_n(\CC)$. 

Recall that by ``subalgebra'' we always mean a $^*$-subalgebra containing the 
identity, so $\A \cap \B$ is at least one-dimensional. Thus
$c(\A,\B)$ is a nonnegative real and infact
\begin{equation}
1\leq c(\A,\B) \leq {\rm min}\{{\rm dim}(\A), {\rm dim}(\B)\}
\end{equation}
with $c(\A,\B)=1$ if and only if $\A$ and $\B$ are quasi-orthogonal.

We are interested by the relation between the quantities 
$c(\A,\B)$ and $c(\A',\B')$ where $\A'$ and $\B'$ are the commutants of $\A$ and
$\B$, respectively. The next example shows that in general, $c(\A',\B')$ cannot 
be determined by the value of $c(\A,\B)$ and the (unitary) equivalence 
classes\footnote{Two isomorphic subalgebras of $M_n(\CC)$ (that is: 
$^*$-subalgebras containing $\mathbbm 1\in M_n(\CC)$) are are not 
necessarily unitarily equivalent. For example, if $P$ and $Q$ are 
ortho-projections in $M_4(\CC)$ onto subspaces of dimensions $2$ and $3$, 
respectively, then $\A \simeq \B \simeq \CC^2$ where $\A=\CC P + \CC \mathbbm 1$ 
and $\B=\CC Q + \CC \mathbbm 1$. However, clearly there is no unitary $U\in 
M_n(\CC)$ such that $U\A U^*$ would coincide with $\B$.} of the subalgebras $\A$ and $\B$.
\bigskip

\noindent
{\it Example.} Let $\A:=M_4(\CC)\otimes \mathbbm 1_4 \subset M_4(\CC)\otimes 
M_4(\CC)\equiv M_{16}(\CC)$ and $\tilde{\A}:=\A'=\mathbbm 1_4 \otimes M_4(\CC)$. 
Then clearly, $\A$ and $\tilde{\A}$ are unitarily equivalent.
Let further $P_1\in M_4(\CC)$ be an ortho-projection onto a 
one-dimensional subspace and $P_2\in M_4(\CC)$ an ortho-projection onto a
two-dimensional subspace. Finally, let $\B:=\D_1 \otimes \D_2$ where 
$\D_1,\D_2\subset M_4(\CC)$ are the abelian subalgebras generated by the single 
ortho-projections $P_1$ and $P_2$, respectivly. Then, as all of the 
algebras $\A, \A', \tilde{\A}, \tilde{\A}', \B,\B'$ have a product-form, it 
is easy to see that 
\begin{equation}
E_\A E_\B = (id_4 \otimes {\rm Tr}_4) (E_{\D_1}\otimes 
E_{\D_2}) = E_{\D_1}\otimes {\rm Tr}_4 = E_{\D_1\otimes \mathbbm 1_4},
\end{equation}
 so 
$c(\A,\B) = {\rm Tr}(E_\A E_\B) = {\rm Tr}(E_{\D_1\otimes \mathbbm 1_4}) = 
{\rm dim}(\D_1) = 2$
and similarly, that $c(\tilde{\A},\B)$ is also equal to $2$. However, as 
$\B'=\D_1'\otimes \D_2'$ while the commutants of $\A$ and $\tilde{\A}$ 
are $\tilde{\A}$ and $\A$, respectively, we have that
\begin{equation}
c(\A',\B') = {\rm dim}(\D_2') = 2^2 + 2^2 \neq 1^2 + 3^2 = {\rm dim}(\D_1') = 
c(\tilde{\A},\B').
\end{equation}
Thus the value of $c(\A,\B)$, even together with the knowledge of the unitary 
equivalence classes of $\A$ and $\B$, is insufficient for determining $c(\A',\B')$.
\bigskip

\noindent
A subalgerba, up to unitary equivalence, is always of the form
\begin{equation}
\A= \oplus_{k} \left(
M_{n_k}(\CC) \otimes \mathbbm 1_{m_k}
\right) \, \subset M_n(\CC) 
\end{equation}
where $n=\sum_k n_k m_k$
(and $\mathbbm 1_x$ is the unit of $M_x(\CC)$)
with commutant
\begin{equation}
\A'= \oplus_{k} \left(\mathbbm 1_{n_k}(\CC) \otimes M_{m_k}(\CC)
\right)\, \subset M_n(\CC).
\end{equation}
In case the ratios $n_k/m_k$ are independent of the index $k$,
we shall say that the subalgebra $\A$ is {\bf homogeneously balanced}. 
Note that if $n_k/m_k = \lambda$ for all indices $k$, then
$n= \sum_k n_k m_k = \lambda \sum_k n_k^2 =\lambda {\rm dim}(\A)$ and
${\rm dim}(\A) = \sum_k n_k^2 = \lambda^2 \sum_k m_k^2 = \lambda^2 {\rm dim}(\A')$.
Some evident, but important consequences 
of our definition and this last remark are:
\begin{itemize}
\item
$\A$ is homogeneously balanced  if and only if so is $\A'$,
\item
if $\A\subset M_n(\CC)$ and $\B\subset M_m(\CC)$ are homogeneously balanced 
then so is the tenzorial product  
$\A\otimes \B \subset M_n(\CC)\otimes M_m(\CC) \equiv M_{nm}(\CC)$,
\item
factors and MASAs are automatically homogeneously balanced,
\item
if $\A$ is homogeneously balanced then ${\rm dim}(\A){\rm dim}(\A') = n^2$
\item
this homogeneity, in general, is not only a condition about the isomorphism class of 
$\A$, but also a condition about the way it ``sits'' in $M_n(\CC)$: of two 
isomorphic subalgebras, one may be homogeneously balanced while the other is not.
\end{itemize}
Note that the algebra $\B$ in the previous example was {\it not} homogeneously balanced.
We shall now recall a simple, but important fact; for the more general statement and its proof
see \cite[Prop.\! 1]{ohnopetz}.
\begin{lemma}
If $\A\subset M_n(\CC)$ is homogeneously balanced and
$A_1,\ldots A_N$ is an 
ortho-normed basis of $\A$, then 
$E_\A(X) = \frac{n}{N}\sum_{j=1}^N A_k X A_k^*$
for all $X\in M_n(\CC)$.
\end{lemma}

\begin{theorem}
If $\A,\B\subset M_n(\CC)$ are homogeneously balanced then
$$
c(\A',\B') = \frac{n^2}{{\rm dim}(\A){\rm dim}(\B)}c(\A,\B).
$$
\end{theorem}
\begin{proof}
Let  $A_1,\ldots A_N\in \A$ and $B'_1,\ldots B'_{\tilde{N}}\in \B'$ be two
ortho-normed bases, where $N:={\rm dim}(\A)$ and 
$ \tilde{N}:={\rm dim}(\B')= n^2/{\rm dim}(\B)$. Using the previous lemma 
\begin{eqnarray}
\nonumber
\sum_{j,k}{\rm Tr}(A_j B'_k A_j^* {B'_k}^*)
 &=& \frac{n}{N}
\sum_k {\rm Tr}(E_{\A'}(B'_k){B'_k}^*) 
= \frac{n}{N}\sum_k 
{\rm Tr}(E_{\A'}(E_{\A'}(B'_k){B'_k}^*))
\\
\nonumber 
&=& \frac{n}{N}\sum_k
{\rm Tr}(E_{\A'}(B'_k)E_{\A}({B'_k}^*))
= \frac{n}{N}\sum_k \| E_{\A'}(B'_k) \|_{{\rm Tr}}^2
\\
\label{eq:ccomp}
&=&\frac{n}{{\rm dim}(\A)} Tr (E_{\A'}E_{\B'})
\end{eqnarray}
where we have used the simple fact that if $P,Q$ are ortho-projections then 
${\rm Tr}(PQ)= \sum_k \|Pq_k\|^2$ where $q_1,\ldots, q_s$ is an ortho-normed basis 
in the image of $Q$. However, we could have carried out the above 
calculation in a similar way but with the role of $\A$ and $\B'$ exchanged. 
Confronting the obtained form to the one appearing in the previous equation one can 
easily obtain the claimed formula.
\end{proof}

\section{Example applications of the formula}
\label{sec:application}

The so-far presented results relied on the result of Petz and Ohno which ensured
that if $\A,\B\subset M_n(\CC)$ are quasi-orthogonal and 
${\rm dim}(\A){\rm dim}(\B) = n^2$ then also $\A'$ and $\B'$ form a quasi-orthogonal pair. 
Note that if $\A$ and $\B$ are homogeneously balanced, then this fact is a simple consequences of our formula obtained in the last section. In some sense, our formula gives a {\it quantitative generalization} of this fact.  
To make use of this quantitative information, all we need is the following observation.

\begin{lemma}
\label{appliformula}
Let $\C\subset M_n(\CC)$ be a subalgebra and $\A_1,\ldots,\A_k\subset M_n(\CC)$ be a system of quasi-orthogonal subalgebras. Then
\begin{itemize}
\item[(i)]
 ${\rm dim}(\C)\geq 1-k + \sum_{j=1}^k  c(\A_j,\C)$,  and
 \item[(ii)]
 ${\rm dim}(\C) \leq  n^2-1+ \sum_{j=1}^k  (c(\A_j,\C)-{\rm dim}(\A_j))$
\end{itemize}
with equality holding in (i) if and only if $\C\subset \A_1+\A_2+\ldots + \A_k$ (which is automatically satisfied if in particular 
$\A_1,\ldots,\A_k$ is a quasi-orthogonal system of $M_n(\CC)$).
\end{lemma}
\begin{proof}
The subalgebra $\CC\mathbbm 1$ is contained in every subalgebra, so 
$E_{\A_j}-E_{\CC\mathbbm 1}$ is a projection; in fact it is the ortho-projection onto
the ``traceless part'' of $\A_j$.  Quasi-orthogonality of $\A_1,\ldots \A_k$ is then equivalent to the fact 
that the projections $(E_{\A_j}-E_{\CC\mathbbm 1})$ are mutually orthogonal for $j=1,\ldots ,k$. Moreover, in this case 
\begin{equation}
F:=E_{\CC\mathbbm 1} + \sum_{j=1}^k (E_{\A_j}-E_{\CC\mathbbm 1}) 
\end{equation}
is nothing else than the ortho-projection onto the subspace $\A_1+\ldots+\A_k$. Hence
\begin{eqnarray}
\label{eq:FE_C}
\nonumber
{\rm Tr}(FE_{\C}) &=& {\rm Tr}(E_{\CC\mathbbm 1}E_\C) + \sum_{j=1}^k
({\rm Tr}(E_{\A_j}E_{\C}) - {\rm Tr}(E_{\CC\mathbbm 1}E_\C) = 
\\
& =&
1 +  (\sum_{j=1}^k  c(\A_,\C) - 1) = 1-k + \sum_{j=1}^k  c(\A_j,\C)
\end{eqnarray}
as $E_{\CC\mathbbm 1}E_\C = E_{\CC\mathbbm 1}$  is a projection onto
a one-dimensional subspace. Thus (i) follows as
\begin{equation}
{\rm Tr}(FE_\C)\leq {\rm Tr}(E_\C) = {\rm dim}(\C)
\end{equation}
with equality holding if and only if $E_\C$ is a smaller projection than $F$; i.e.\! when 
$\C\subset (\A_1+\ldots + \A_k)$. 
The inequality (ii) follows by considering $\mathbbm 1$
as the sum of the two orthogonal projections: $\mathbbm 1 = F + (\mathbbm 1-F)$. Now
$F$ is an ortho-projection onto a $d:={\rm dim}(\sum_{j}\A_j)$ dimensional space, 
where by quasi-orthogonality
\begin{equation}
d = 1+\sum_j({\rm dim}(A_j)-1),
 \end{equation}
whereas $(\mathbbm 1-F)$ is the ortho-projection onto the orthogonal of $\A_1+\ldots+\A_k$, which is an 
$n^2-d$ dimensional subspace. Thus
\begin{equation}
{\rm dim}(\C)={\rm Tr}(E_\C) = {\rm Tr}(FE_\C)+ {\rm Tr}((\mathbbm 1-F)E_\C)
\end{equation}
where ${\rm Tr}((\mathbbm 1-F)E_\C)\leq n^2-d =n^2 -1-\sum_j({\rm dim}(A_j)-1)$. This, together
with (\ref{eq:FE_C}) expressing the term ${\rm Tr}(FE_\C)$, concludes our proof.
\end{proof}

So let us see now how we can use our formula in practice. We begin with a fairly 
simple case; we shall 
consider a quasi-orthogonal system in $M_6(\CC)$ containing $6$ maximal abelian subalgebras 
$\A_1,\ldots \A_6$ and a subalgebra $\B$ isomorphic to $M_2(\CC)$. 
The trivial necessary conditions would allow the existence of such a system. 
Of course, as was explained in the remark made after 
Corollary \ref{1factor}, by using the strong result of \cite{lmsweiner}, it is easy to show that such a system cannot exists. 
But how could we rule out its existence in a more direct manner? 
Now we cannot use  Corollary \ref{1factor}: $6$ is not a square number and more in 
particular ${\rm dim}(M_2(\CC))= 2^2\neq 6$, so $\B'$ would not remain quasi-orthogonal to 
the subalgebras $\A_1,\ldots, \A_6$. 

So assume the existence of such a system. Then by the fact that $\A_j$ is a MASA $(j=1,\ldots ,6)$, and 
by an application of our formula 
\begin{equation}
c(\A_j,\B') = c(\A_j',\B') = \frac{6^2}{6*4}c(\A_j,\B) = \frac{3}{2},
\end{equation}
since ${\rm dim}(\B)= {\rm dim}(M_2(\CC)) =4$ and $c(\A_j,\B) = 1$ by the assumed quasi-orthogonality
of $\B$ and $\A_j$.  Moreover, $c(\B,\B')=1$ because $\B$ was assumed to be a factor.
So considering the quasi-orthogonal system $\A_1,\ldots,\A_6,\B$ and the algebra $\C:=\B'\simeq M_3(\CC)$, we have 
${\rm dim}(\A_j) = 6, \, {\rm dim}(\B) = 2^2 = 4$, and 
\begin{eqnarray}
\nonumber
n^2-1 + (c(\B,\C)-{\rm dim}(\B))
+   \sum_{j=1}^6 (c(\A_j,\C)-{\rm dim}(A_j))
&= &
\\
6^2-1 \; +\;\;\;\;\;\;\;\;(1-4)\; \;\;\;\;\;\;\;\;\; +\;\;\; \;\;\;\; 6 \;*\;\;(\frac{3}{2}-6)\;\;\; \;\;\;\;\;\;\;\;&=& 5,
\end{eqnarray}
which is in conflict with (ii)  of lemma  \ref{appliformula}, as
${\rm dim}(\C)={\rm dim}(\B')=3^2 =9\nleq 5$.

This is nice, but --- as was mentioned --- it is a fairly simple case in which we have already known the nonexistence. So we shall finish by considering a somewhat more complicated example.

\begin{proposition}
There is no quasi-orthogonal system in $M_6(\CC)$ consisting of $5$ maximal abelian subalgebras  and $3$ factors isomorphic to $M_2(\CC)$.  
\end{proposition}
\begin{proof}
Again, note that the existence of such a system cannot be ruled out by the trivial necessary conditions. 
We assume $\A_1,\ldots,\A_5,\B_1,\B_2\B_3$ is such a system (with the ``$\A$'' algebras being the maximal 
abelian ones, and the ``$\B$'' algebras the factors isomorphic to $M_2(\CC)$).
To apply our formula, we will need to consider the commutants as well as the original algebras. 
If $\B_1,\B_2,\B_3\simeq M_2(\CC)$ and $\B_1,\B_2,\B_3\subset M_6(\CC)$, then 
their commutants $\B'_1, \B'_2, \B'_3 \simeq M_3(\CC)$.  Since $M_2(\CC)$ cannot be embedded in $M_3(\CC)$ in an identity preserving way, 
we have that $\B_j$ is not contained in $\B'_k$ and consequently that 
\begin{equation}
 c(\B_j,\B'_k) < {\rm dim}(\B_j) = 4
 \end{equation}
for every $j,k=1,2,3$. However, we shall need a better estimate. The fact is that $\B_j$ is not only not contained in $\B'_k$, but actually
we can say something about their ``minimal distance''. We shall shortly interrupt our proof with a lemma concerning this issue.
\begin{lemma}
$c(\B_j,\B'_k) \leq 3$.
\end{lemma}
\noindent
{\it Proof (of lemma).}
Let $X,Y,Z, W\in \B_j$ be an orthogonal basis such that $W=\mathbbm 1$ and
$X,Y,Z$ correspond to the Pauli-matrices in a suitable identification $\B_j\simeq M_2(\CC)$.
Let us further denote the trace-preserving expectation onto $\B'_k$ by $E$. Then 
$E(X)$ (and similarly $E(Y)$ and $E(Z)$, too) remains self-adjoint, 
so it is unitarily equivalent with a diagonal matrix. Moreover, 
as it belongs to $\B'_k\simeq M_3(\CC)$, we may actually assume it 
is unitarily equivalent with the diagonal matrix  
 ${\rm diag}(\lambda_1,\lambda_1,\lambda_2,\lambda_2,\lambda_3,\lambda_3)\in M_6(\CC)$. 
 We have that
\begin{itemize}
 \item $\lambda_1+ \lambda_2 + \lambda_3 = 0$,
 \item $\lambda_1,\lambda_2,\lambda_3 \in [-1,1]$.
 \end{itemize}
 Indeed,  the first equation follows as ${\rm Tr}(E(X))= {\rm Tr}(X)=0$, whereas the second follows from the fact 
 $\mathbbm 1 \pm E(X) = E(\mathbbm 1 \pm X)$ --- just as $\mathbbm 1 \pm X$ --- is a positive operator.
Now elementary calculus shows that in the region determined by the two equation, we have 
\begin{equation}
{\rm Tr}(E(X)^2) = 2(\lambda_1^2 + \lambda_2^2 + \lambda_3^2) \leq 4. 
\end{equation}
 On the other hand, ${\rm Tr}(X^2) = {\rm Tr}(\mathbbm 1)=6$; actually, 
 $X,Y,Z,\mathbbm 1$ is an orthogonal basis whose each member has 
 (trace)norm-square equal to $6$. Thus, using the arguments explained at 
 and after equation (\ref{eq:ccomp}),  we have that
\begin{equation}
c(\B_j,\B'_k) = 
\frac{1}{6}({\rm Tr}(E(X)^2+E(Y)^2+E(Z)^2+E(\mathbbm 1)^2))
\leq \frac{1}{6}(4 + 4+ 4+ 6) = 3 
\end{equation}
which is just what we wanted to prove.
 \quad
 \smallskip

To finish the proof, we consider the algebra $\C:=\B'_1\simeq M_3(\CC)$ and the quasi-orthogonal system 
$\A_1,\ldots , \A_5,\B_1,\B_2,\B_3$. As $\A'_k=\A_k$. By an application of our formula have that 
\begin{equation}
c(\A_k,\C) = c(\A_k,\B'_1) = \frac{6^2}{6*2^2 } c(\A_k,\B_1) = \frac{6^2}{2^2*6} = \frac{3}{2}
\end{equation}
since $c(\A_k,\B_1)=1$ by quasi-orthogonality. Now
$c(\B_1,\C) = c(\B_1,\B'_1)=1$ since
$\B$ is a factor, and finally, for $c(\B_2,\C)$ and $c(\B_3,\C)$ we can use the estimate provided 
by the lemma we have just made.
To sum it up: we have $n^2-1=6^2-1=35$,
\begin{eqnarray}
\nonumber
 \sum_{j=1}^5(c(\A_j,\C) - {\rm dim}(\A_j)) &= & 5 * (\frac{3}{2}-6) = -\frac{45}{2}, \;\;\;\; {\rm and}
\\
 \sum_{j=1}^3(c(\B_j,\C) - {\rm dim}(\B_j)) &\leq& (1-4) + (3-4)  +(3-4) = - 5
 \end{eqnarray}
 which gives $35-(45/2)-5 = 15/2 \ngeq 9  ={\rm dim}(\C)$, in contradiction with point (ii) of lemma
 \ref{appliformula}.
\end{proof}


\begin{thebibliography}{99}

\bibitem{h6}
I.\! Bengtsson, W.\! Bruzda, \AA.\! Ericsson, J.-A.\! Larsson,
W.\! Tadej and K.\! Zyczkowski:
Mutually unbiased bases and Hadamard matrices of order six,
{\it J.\! Math.\! Phys.} {\bf 48} (2007), 052106.

\bibitem{numerical}
P.\! Butterley and W.\! Hall:
Numerical evidence for the maximum number of mutually
unbiased bases in dimension six,
{\it Phys.\! Lett.\! A} {\bf 369} (2007), 5.

\bibitem{security}
N.\! J.\! Cerf, M.\! Bourennane, A.\! Karlsson and   
N.\! Gisin: Security of quantum key distribution using $d$-level
systems, {\it Phys.\! Rev.\! Lett.} {\bf 88} (2002), 127901.

\bibitem{woofie}
B.\! D. Fields and W.\! K.\! Wootters:
Optimal state-determination by mutually unbiased measurements,
{\it Ann.\! Phys.} {\bf 191} (1989), 363.

\bibitem{ivanovic}
I. D. Ivanovi\'c: Geometrical description of
quantal state determination, {\it J.\! Phys.\! A} {\bf 14} (1981), 3241.

\bibitem{matemisi}
P.\! Jaming, M.\! Matolcsi, P.\! M\'ora, F.\! Sz\"oll\H{o}si and M.\!
Weiner: A generalized Pauli problem and an infinite family of 
MUB-triplets in dimension 6, {\it J.\! Phys.\! A} {\bf 42} (2009), 
no. 24, 245305.

\bibitem{ohno}
H.\! Ohno: Quasi-orthogonal subalgebras of matrix algebras,
{\it Linear Alg.\! Appl.} {\bf 429} (2008), 2146.

\bibitem{ohnopetz}
H.\! Ohno and D.\! Petz: Generalizations of Pauli
channels. {\it Acta Math. Hungar.}, {\bf 124} (2009), 165.

\bibitem{opsz}
H.\! Ohno, D.\! Petz and A.\! Sz\'ant\'o:
Quasi-orthogonal subalgebras of $4 \times 4$ matrices,
{\it Linear Alg.\! Appl.} {\bf 425} (2007), 109.

\bibitem{petz} D.\! Petz,
Complementarity in quantum systems,
{\it Rep.\! Math.\! Phys.} {\bf  59} (2007), 209.

\bibitem{petzkhan}
D.\! Petz and J.\! Kahn:
Complementary reductions for two qubits,
{\it J.\! Math.\! Phys.} {\bf 48} (2007) 012107.

\bibitem{pszw}
D.\! Petz, A.\! Sz\'ant\'o and M.\! Weiner:
Complementarity and the algebraic structure of 4-level quantum
systems, {\it J.\! Infin.\! Dim.\! Anal.\! Quantum Probability and Related Topics} {\bf 12} 
(2009),  99.

\bibitem{lmsweiner}M.\! Weiner: A gap for the maximum number of
mutually unbiased bases. {\it Under peer-review},
{\tt arXiv:0902.0635}.

\bibitem{wooters}
W.\! K.\! Wootters: A Wigner-function formulation
of finite-state quantum mechanics, {\it Ann.\! Phys.}
{\bf 176} (1987), 1.

\end{thebibliography}
\end{document}